\documentclass[conference]{IEEEtran}
\IEEEoverridecommandlockouts
\usepackage{cite}
\usepackage{amsmath,amssymb,amsfonts}
\usepackage{algorithmic}
\usepackage{graphicx}
\usepackage{textcomp}
\usepackage{xcolor}
\usepackage{flushend}
\def\BibTeX{{\rm B\kern-.05em{\sc i\kern-.025em b}\kern-.08em
    T\kern-.1667em\lower.7ex\hbox{E}\kern-.125emX}}

\usepackage{amsthm}
\newtheorem{assumption}{Assumption}
\newtheorem{proposition}{Proposition}
\newtheorem*{remark}{Remark}
\newtheorem{theorem}{Theorem}
\newtheorem{lemma}{Lemma}

\begin{document}

\title{Sum-of-Squares Program and Safe Learning On Maximizing the Region of Attraction of Partially Unknown Systems

}

\author{\IEEEauthorblockN{Dongkun Han}
\IEEEauthorblockA{\textit{Department of Mechanical and Automation Engineering} \\
\textit{The Chinese University of Hong Kong}\\
HKSAR, China \\
dkhan@mae.cuhk.edu.hk}
\and
\IEEEauthorblockN{Hejun Huang}
\IEEEauthorblockA{\textit{Department of Mechanical and Automation Engineering} \\
\textit{The Chinese University of Hong Kong}\\
HKSAR, China \\
hjhuang@mae.cuhk.edu.hk}
}

\maketitle

\begin{abstract}
Recent advances in learning techniques have enabled the modelling of unknown dynamical systems directly from data. However, in many contexts, these learning-based methods are short of safety guarantee and strict stability verification. To address this issue, this paper first approximates the partially unknown nonlinear systems by using a learned state space with Gaussian Processes and Chebyshev interpolants. A Sum-of-Squares Programming based approach is then proposed to synthesize a controller by searching an optimal control Lyapunov Barrier function. In this way, we maximize the estimated region of attraction of partially unknown nonlinear systems, while guaranteeing both safety and stability. It is shown that the proposed method improves the extrapolation performance, and at the same time, generates a significantly larger estimated region of attraction.
\end{abstract}


		\section{Introduction}
		
        Consider a Sum of Squares (SOS) polynomial $p(x)$ as
		\begin{equation}
		    \label{eqn:openning}
		    p(x)=\sum^n_{i=1}f_i^2(x),
		\end{equation}
		\noindent where $f_i(x)$, $i=1,\dots,n$ is a monomial, the nonnegativeness of $p(x)$ could be directly obtained: $\forall x\in \mathbb{R}^n,~p(x)\geq 0$. Checking whether a polynomial is an SOS boils down to a semidefnite program which turns out to be a convex optimization and solvable in polynomial time. This so-called SOS program has been successfully used to regulate the stability and safety of control-affine systems \cite{chesi2011domain,majumdar2013control}. Various SOS related results have been derived with different practical concerns, like sparse computation \cite{zheng2019sparse} and clarity identification \cite{ahmadi2020complexity1}. However, these SOS techniques \cite{jarvis2005control} are not easy to cope with nonpolynomial terms and partial unknown dynamics in real world, such as in fluid mechanics \cite{najm2009uncertainty} and biological systems \cite{papachristodoulou2005analysis}. 
		
		To address above shortcomings, tools, such as SOS robust controller, have been developed to handle parametric uncertainty \cite{dorobantu12SSDM,han2016estimating, han2014tcas1,han19tac,han12tii}. Meanwhile, learning-based methods have strengthened the link between optimizations and dynamics modelling. Many recent approaches exploit Gaussian Process (GP) with limited prior knowledge \cite{umlauft2017learning,berkenkamp2016safe,buisson2020actively}. In contrast to parametric uncertainties, these learning-based methods rarely provide safety guarantees, which limits the range of their real-world applications.
		
		A number of pioneering explorations have been attempted to guarantee the safety of above learning-based methods. Reinforcement learning approaches have obtained a safe exploration inside the ROA \cite{berkenkamp2017safe,jin2018control,cheng2019end} and data-driven learning approaches can provide a rich set of control options at a certain level of probability \cite{wang2018safe,hewing2020learning}. Among these methods, the ROA certified by the Control Lyapunov Function (CLF) is shown to be a useful approach. \cite{ames2019control} define a type of dynamical safety which could be certified by the Control Barrier Function (CBF). Based on this idea, \cite{wang2018permissive} introduces an SOS-based strategy to construct permissive barrier certificates for estimating the safe regions. However, when the model confronts of disturbances or non-polynomial nonlinearity, this method can hardly find a solution. Based on \cite{wang2018permissive,devonport2020bayesian}, our previous work \cite{huang21ccdc} relaxes polynomial restriction and the constraints of polynomial kernel in representations, and proposes a method to approximate the partially unknown nonlinear autonomous systems into a polynomial form.
		
		In contrast to the aforementioned methods, this paper interprets the safety and stability properties in a more general form based on a CLF and CBF certified ROA, which is able to obtain a better estimation result. The main contributions of this paper are threefold. First, for partially unknown systems, we reconstruct a learned control affine system through Gaussian Processes and Chebyshev interpolants. SOS-based sufficient conditions are proposed for the existence of a feasible controller such that the safety and stability of learned systems can both be guaranteed (Theorem 1). Second, solvable conditions have provided for searching an optimal control Lyapunov barrier function without and with considering pre-defined unsafe regions (Theorem 2 and Theorem 3). Finally, an algorithm is developed to search for the optimal controller such that the estimated ROA can be maximized (Section IV). Numerical examples demonstrate that a significantly larger ROA can be obtained by the proposed method (Section V).
		
		\section{Preliminary} \label{sec:prelimary}
		
		Notation: Let $\mathcal{P}$ be the set of polynomials and $\mathcal{P}^\text{SOS}$ be the set of sum of squares polynomials, i.e., $P(x)=\sum_{i=1}^{k}p_i^2(x), $ where $P(x)\in \mathcal{P}^{SOS}$ and $p_i(x)\in\mathcal{P}$. 
		
		Consider a control affine dynamical system as follows,
		\begin{equation}
			\label{eqn:sysstate}
			\begin{aligned}
				\dot{x} &= f(x)+g(x)u+d(x),\\
			\end{aligned}
		\end{equation}
		\noindent where $x \in \mathcal{X} \subset \mathbb{R}^{n}$ and $u \in \mathcal{U} \subset \mathbb{R}^{m}$ denote the state and control input of the system, which has three Lipschitz continuous terms, $f:\mathbb{R}^{n}\rightarrow \mathbb{R}^{n}$ denotes a nonlinear term, $g:\mathbb{R}^{n}\rightarrow \mathbb{R}^{n\times m}$ denotes a polynomial term and $d:\mathbb{R}^{n}\rightarrow \mathbb{R}^{n}$ denotes an unknown term. In this paper, we consider a control input $u$ in a polynomial form \cite{chesi2011domain,papachristodoulou2005analysis}.
		
		Our prior information of the system (\ref{eqn:sysstate}) come from the measurements in the form $(x,u,\dot{x})$. Normally, $\dot{x}$ is a finite-difference approximation term based on the neighbouring $x$ that in itself is not directly measurable, while the unmodeled term $d(x)$ describes all the noise around the estimate of $\dot{x}$, as well as the measurements' inaccuracy of $x$ in practice. The measurements of $d(x)$ can be obtained by subtracting $f(x)+g(x)u$ from $\dot{x}$. 
		
		We assume that $d(x)$ is an independent zero-mean Gaussian noise which exists across all measurements. For simplicity, all components of $d(x)$ assume to have the same variance $\sigma_n^2$.
		\begin{assumption}\label{asm:noisedis}
			The unknown term $d(x)$ is a zero-mean Gaussian noise and uniformly bounded by $\sigma_n$.  $\hfill\square$
		\end{assumption}
		
		We further assume $d(x)$ is bounded in the system (\ref{eqn:sysstate}):
		\begin{assumption}\label{asm:boundnorm}
			The unknown term $d(x)$ in (\ref{eqn:sysstate}) exists a bounded norm in the reproducing kernel Hilbert space, $\forall x\in \mathcal{X}, \Vert d(x)\Vert\leq c_g$, where $c_g$ is a constant. $\hfill\square$
		\end{assumption}	
		
		\subsection{Model Formulation}
		In this subsection, we make dynamical approximations of (\ref{eqn:sysstate}) to obtain a probabilistic polynomial representation by Chebyshev interpolants and Gaussian processes regression, respectively.
		
		\subsubsection{Chebyshev Interpolants}
		
		
		Chebyshev interpolants (CI) provide an important method to capture complex functions via Chebyshev polynomials in $[-1, 1]$. Note that, CI can work in any arbitrary interval $[a,b]$ through the transformation \cite{trefethen2019approximation} as follows,
		
		\begin{equation}
			\begin{aligned}
				\hat{x} = \frac{x-1/2(b+a)}{1/2(b-a)}.
			\end{aligned}
		\end{equation}
		The Chebyshev polynomial of degree $k$ is given as, $\forall i \in [0,k],T_i(x)=\cos(i \arccos(x))$, which satisfies a recursion $T_0(x)=1, T_1(x)=x, T_{i}(x)=2xT_{i-1}(x)-T_{i-2}(x)$. The corresponding value of $k+1$ Chebyshev nodes $x_i=\cos (i\pi / k)$ at the target $f(x)$ allow us to obtain an approximation
		\begin{equation}\label{eqn:cheby}
			f(x)\approx P_k(x)=\sum^k_{i=0}c_i T_i(x),
		\end{equation} 
		\noindent where $c_i=\frac{2}{\pi}\int_{-1}^1(1-x^2)^{-1/2}f(x)T_i(x)dx$. We define that $\xi(x)= f(x)-P_k(x)$ from (\ref{eqn:cheby}) to obtain a bounded remainder of CI shown in the following result. 
		
		\begin{lemma}(Theorem 8.2 of \cite{trefethen2019approximation}) 
			\label{lm:Approximation}
			Let a bounded function $\vert f(x)\vert \leq c_m$ analytic in $[-1,1]$ be analytically containable to the open Bernstein ellipse $E$ where $f(x)$ is a non-polynomial term in (\ref{eqn:sysstate}) and $c_m$ is a constant. Then, the upper bound of the remainders from Chebyshev interpolants $P_k(x)$ of degree $k\geq 0$ satisfies
			\begin{eqnarray}\label{eqn:Chebyshev}
				\begin{aligned}
					\Vert f(x)-P_k(x) \Vert \leq \frac{4c_m  \rho^{-k}}{\rho -1}.
				\end{aligned}
			\end{eqnarray}
			\noindent The Bernstein ellipse $E$ has foci $\pm 1$ and major radius $1+\rho$ for all $\rho \ge 0$ containing the concerned region. $\hfill\square$
		\end{lemma}
		
		
		\noindent Let $\xi(x) \in [ -\frac{4c_m  \rho^{-k}}{\rho -1},\frac{4c_m  \rho^{-k}}{\rho -1}]$ so that an equivalent statement of (\ref{eqn:sysstate}) can be obtained as follows. 
		
		\begin{proposition}
			The partially unknown system (\ref{eqn:sysstate}) can be approximated by the Chebyshev interpolants in a certain region as 
			\begin{equation}\label{eqn:approsys}
				\begin{aligned}
					\dot{x}=P_k(x)+g(x)u+d(x)+\xi(x),
				\end{aligned}
			\end{equation}
			which satisfies Assumption \ref{asm:noisedis} and \ref{asm:boundnorm}, simultaneously.
		\end{proposition}
		\begin{proof}
			The value of $\xi(x)$ does not obey any distribution so that the unknown term $d(x)+\xi(x)$ satisfies Assumption \ref{asm:noisedis}. Next, the inequality (\ref{eqn:Chebyshev}) from \cite{trefethen2019approximation} declares that the Chebyshev interpolants can be used in an arbitrary region while $\xi(x)$ in this region is bounded. Thus, $d(x)+\xi(x)$ is still bounded in the region, which completes this proof. 
		\end{proof}
		
		\subsubsection{Gaussian Processes}
		
		Gaussian processes (GP) is able to capture the unmodeled dynamics and further infer the dynamical information based on the prior knowledge. With GP, every element inside the finite subset $\{x_1,x_2,\cdots,x_k\}\in\mathcal{X}$ is associated with a normally distributed random variable. Thus, the prior and posterior model of GP obeys a joint Gaussian distribution over the function of these variables. More specifically, it is considered to describe this distribution of the target function $\hat{f}:\mathcal{X}\rightarrow \mathbb{R}$ by GP as 
		\begin{equation}
			\begin{aligned}
				\hat{f}\sim \mathcal{GP}(m(x),k(x,x^{\prime})),
			\end{aligned}
		\end{equation}
		
		\noindent where $m(x)$ is the mean function of this distribution and usually equals to zero in the prior model, $k(x,x^{\prime})$ is the kernel function to measure the similarity of the horizon states $x,x^{\prime}\in \mathcal{X}$. The GP model is synthesized by the hyper-parameter $\Omega$ which characterizes the value of distribution $\hat{f}$. A widely used kernel function is the isotropic squared exponential kernel
		
		\begin{equation}
			\begin{aligned}
				k^{SE}(x,x^{\prime})=\sigma_f^2 \exp (\sum^n_{j=1}\frac{(x_j-x_j^{\prime})^2}{-2l^2_j}),
			\end{aligned}
		\end{equation}
		
		\noindent where $\sigma_f$ is the signal variance and $l_j$ is the length scale. We are able to obtain a mean function in polynomial form as shown in our previous work.
		
		\begin{lemma}(Proposition 1 of \cite{huang21ccdc})
			\label{lem:polynomial_mean}
			A GP posterior model can be represented by self-defined polynomial mean functions $m(x_*)$ at the query state $x_*\in\mathbb{R}^n$.
		\end{lemma}
		
		Generally, the unknown terms in (\ref{eqn:approsys}) could be learned by GP model, and the mean function could be represented into a polynomial form. It generates a probabilistic bound of learned results as follows.
		
		\begin{lemma}(Theorem 1 of \cite{huang21ccdc})
			\label{lem:polynomial_model}
			Consider a partially unknown but asymptotically stable system $\dot{x}=f(x)+d(x)$, where $f(x)$ is a given nonlinear function and be approximated by $k^{th}$ degree Chebyshev interpolants $P_k(x)$ in a certain region, and $d(x)$ is an unknown noise signal which satisfies $(0, \sigma_n^2)$. Set $\delta \in (0,1)$. If $m$ measurements of the unknown term $d(x)+\{f(x)-P_k(x)\}$ are given, a polynomial dynamical system established toward the exact dynamics with probability greater or equal to $(1-\delta)^m$, the ROA of this probabilistic system also has the same probability bounds.
		\end{lemma}
		
		Lemma \ref{lem:polynomial_model} gives an estimate of partially unknown dynamical system in polynomial form with probability bounds. Thus, the estimate of the unknown term $d(x)+\xi(x)$ satisfies
		
		\begin{equation}
			\begin{aligned}
				\mathcal{D}(x)=\{d\vert &m_{d_\xi}(x)-k_\delta \sigma_{d_\xi}(x)\leq d(x)+\xi(x)\\
				&\leq m_{d_\xi}(x)+k_\delta\sigma_{d_\xi}(x)\},
			\end{aligned}
		\end{equation}
		
		\noindent with probability bounds $[(1-\delta)^m,1]$, where $k_\delta$ is a design parameter to obtain the probability bounds $((1-\delta)^n,1)$, $\delta\in(0,1)$. Then, the system (\ref{eqn:approsys}) can be expressed into a probabilistic polynomial form as follows,
		\begin{equation}\label{eqn:learnsys}
			\begin{aligned}
				\dot{x}=P_k(x)+g(x)u(x)+d.
			\end{aligned}
		\end{equation}

		\begin{remark}
			To consider all the possible values of $d\in \mathcal{D}(x)$, we will show in the following section that the controller design takes into account the variance of (\ref{eqn:learnsys}).
		\end{remark}
		
		\subsection{Problem Formulation}
		
		We now introduce a continuous differentiable function $B(x):\mathbb{R}^n \rightarrow \mathbb{R}$ to encode a set $\mathcal{B}$ as,
		\begin{equation}
			\begin{aligned}
				\label{Def:B1}\mathcal{B}=\{x\in \mathcal{X}:B(x)\geq 0\}.
			\end{aligned}
		\end{equation}
		\noindent For safety-critical control, if $\forall x \in \partial \mathcal{B},~\frac{\partial B}{\partial x}\neq 0$ and there exists an extended class $\kappa_\infty$ function $\gamma$ ($\gamma(0)=0$ and strictly increasing) such that for the system (\ref{eqn:learnsys}), $B(x)$ satisfies
		\begin{equation}
			\begin{aligned}\label{eqn:relaxB}
				\exists\; u \; \text{s.t.} \;\frac{\partial{B(x)}}{\partial{x}}(P_k(x)+g(x)u+d)\geq 0,
			\end{aligned}
		\end{equation}
		\noindent where the state trajectory starting inside $\mathcal{B}$ will never access the region $\mathcal{X}\backslash\mathcal{B}$. The controller $u$ from (\ref{eqn:relaxB}) declares the system safety if one can find a set $\mathcal{B}$. Thus, we consider $\mathcal{B}$ as a safe region and the function $B(x)$ is a control barrier function (CBF). 
		
		A typical way to find the region of attraction (ROA) $\mathcal{R}$ is to use the sublevel set of a control Lyapunov function (CLF) $V(x):\mathcal{R}^n\rightarrow \mathcal{R}$:
		\begin{equation}
			\begin{aligned}
				\mathcal{R}=\{x\in\mathcal{X}\backslash0:V(x)> 0\}. 
			\end{aligned}
		\end{equation}
		\noindent In order to drive the state trajectory to the equilibrium point, as a continuously differential function, $V(x)$ satisfies
		\begin{equation}\label{eqn:Lyapunov}
			\begin{aligned}
				\exists \;u \; \text{s.t.}\; \frac{\partial V}{\partial x}(P_k(x)+g(x)u+d)+\zeta\leq 0,
			\end{aligned}
		\end{equation} 
		\noindent where $\zeta$ is a positive scalar to relax this condition. To find a controller guaranteeing safety and stability at the same time, i.e., (\ref{eqn:relaxB}) and (\ref{eqn:Lyapunov}) both hold, an intuitive way is to unify both CLF and CBF as follows,
		\begin{equation}
			\label{eqn:pre-controller}
			\begin{aligned}
				&  && \qquad u^*=  \underset{ u\in\mathbb{R}^+}{\text{argmax}\; J(u)+k_\delta \delta^2} \\
				&\text{s.t.} && -\frac{\partial V(x)}{\partial x}(P_k(x)+d)+\zeta \geq \frac{\partial V(x)}{\partial x}g(x)u(x),\\
				& && \frac{\partial B(x)}{\partial x}(P_k(x)+d)+\gamma(B(x)) \geq -\frac{\partial B(x)}{\partial x}g(x)u(x).
			\end{aligned}
		\end{equation}
		\noindent In this paper, we will compute an optimal controller which enables the largest estimate of ROA where both stabilization and safety constraints strictly hold.  
		
		\section{Main result}
		
		This section attempts to solve the safe stabilization problem described by (\ref{eqn:pre-controller}) for the system (\ref{eqn:sysstate}). Without losing generality,  we assume the equilibrium point of system (\ref{eqn:sysstate}) is the origin as follows.
		
		\begin{assumption}
			\label{asp:1}
			The origin ($x=0$) is a single stable equilibrium of (\ref{eqn:sysstate}). $\hfill\square$
		\end{assumption}

		\subsection{Safe and Stabilization Controller Investigation}\label{sec:controller}
		
		To deal with the nonnegativeness constraints in constructing SOS program, it is quite useful to introduce Positivestellensatz (P-satz) in the following lemma.
		
		\begin{lemma}\label{lem:psatz}
			(\cite{putinar1993positive}) For polynomials $\{a_i\}_{i=1}^m$, $\{b_j\}_{j=1}^n$ and $p$, define a set $\mathcal{B}=\{x\in\mathcal{X}:\{a_i(x)\}_{i=1}^m=0,\{b_j(x)\}_{j=1}^n\geq0\}$. Let $\mathcal{B}$ be compact. The condition $\forall x \in \mathcal{X}, p(x)\geq0$ holds if the following condition holds,
			
			\vspace{7pt}\centering{\hspace{32pt}
				$\left\{
				\begin{array}{l}
					\exists r_1,\dots, r_m \in \mathcal{P}, ~ s_1,\dots, s_n \in \mathcal{P}^{\text{SOS}}, \\
					p-\sum^{m}_{i=1}r_i a_i-\sum^{n}_{j=1}s_j b_j \in \mathcal{P}^{\text{SOS}}.
				\end{array} 
				\right. $}
			\hfill$\square$
		\end{lemma}
		This lemma shows that any strictly positive polynomial $p$ is in the cone that generated by polynomials $\{a_i\}_{i=1}^m$ and $\{b_j\}_{j=1}^n$. Lemma \ref{lem:psatz} provides a useful perspective to satisfy the non-negativity constraint over the SOS programs. It will be adequately used in the following context.
		
		Control barrier functions are developed in this section to compute the optimal controller with a permissive region of attraction, where the system state is both stabilized and in the safe set. We will consider the safe stabilization problem described by (\ref{eqn:pre-controller}) for the system (\ref{eqn:sysstate}) and the learned one (\ref{eqn:learnsys}). Instead of relaxing the stabilization term with $\delta$ in (\ref{eqn:pre-controller}), we will use a bilinear search in SOSP to obtain the permissive ROA.
		
		\begin{theorem}\label{them:1}
			Given the learned control-affine system (\ref{eqn:learnsys}), provided that there exists two local SOS polynomials $s_1(x),s_2(x)$, a sublevel set of the polynomial Lyapunov function $V(x)$ of (\ref{eqn:learnsys}) $\mathcal{L}_V = \{\forall x\in\mathcal{L}_c, V(x)\leq c\}$, an initial barrier function $B(x)=c-V(x)$, the stabilization and safety of (\ref{eqn:learnsys}) can be guaranteed under the control input $u(x)$ generated by the following SOS program:
			\begin{equation}\small
				\begin{aligned}\label{eqn:them1}
					& && \qquad \underset{u(x)\in\mathcal{P};\;\; s_1(x),s_2(x)\in\mathcal{P}^{SOS};\;\; \epsilon\geq 0}{\max \quad \epsilon}\\
					&\text{s.t.} &&-\frac{\partial V(x)}{\partial x}(P_k(x)+g(x)u+d)-s_1(x)B(x) \in \mathcal{P}^{\text{SOS}},\\
					& && \frac{\partial B(x)}{\partial x}(P_k(x)+g(x)u+d)+(\alpha -s_2(x))B(x)-\epsilon \in \mathcal{P}^{\text{SOS}},
				\end{aligned}
			\end{equation}
			\noindent where $\epsilon$ denotes the maximum barrier constraint margin, $\alpha B(x)$ denotes an extended $\kappa$ function. 
		\end{theorem}
		\begin{proof}
			The initial control barrier function can be selected based on the given $V(x)$ in the way that $B(x)=c-V(x)$. It results in a safe region $\{\forall x\in \mathcal{L}_c,\, B(x)=c-V(x)\geq 0\}$, which satisfies (\ref{Def:B1}) and (\ref{eqn:relaxB}). From Lemma \ref{lem:psatz}, the constraints of (\ref{eqn:them1}) imply that (\ref{Def:B1}) and (\ref{eqn:relaxB}) hold in the sublevel set $\mathcal{L}_V$ by using two SOS polynomials $s_1(x) and s_2(x)$. From the SOS program (\ref{eqn:them1}), it generates a control input $u(x)$ such that (\ref{eqn:relaxB}) and (\ref{eqn:Lyapunov}) both hold. In other words, the region $\{\forall x\in \mathcal{L}_c,\, B(x)=c-V(x)\geq 0\}$ guarantees stabilization and safety simultaneously. The non-negative constant $\epsilon$ in the second constraint of (\ref{eqn:them1}) settles a factor in relaxing the control barrier function's constraint (\ref{eqn:relaxB}). Thus, when $\epsilon$ is being maximized, the constraint margin will be expanded in finding the control input $u$. The trajectory $\Psi(t,x_0)$ of any states $x_0\in\mathcal{L}_c$ are always driven to the equilibrium point, which completes the proof.
		\end{proof}	
		
		\subsection{The Largest Safe Region Estimation}
		
		\begin{figure*}[ht] 
			\centering
			\includegraphics[width=1.0\linewidth]{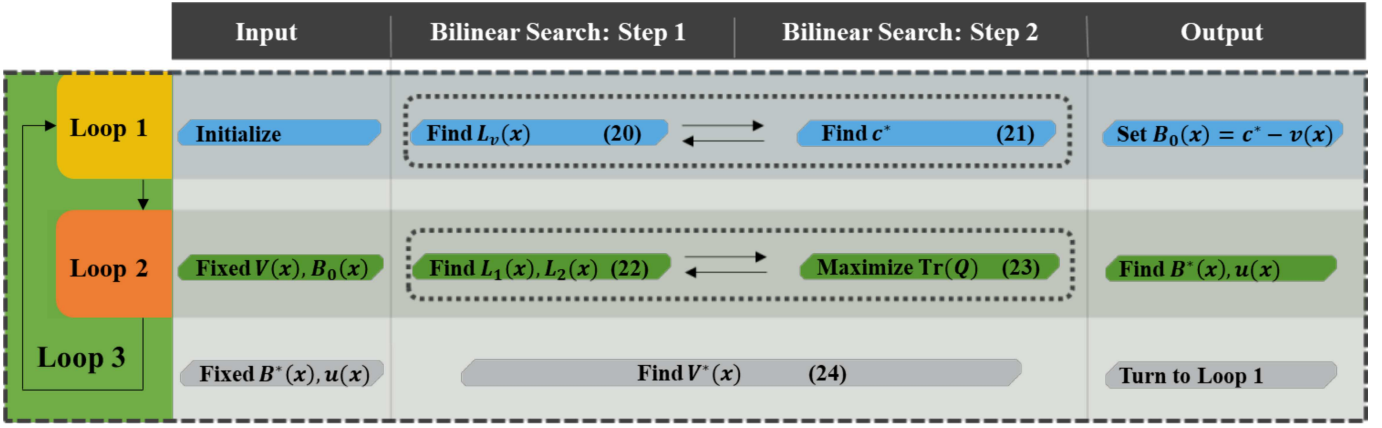}
			\caption{Workflow of the algorithm}
			\label{fig:algorithm}
		\end{figure*}

		The barrier function $B(x)$ can be re-written as a square matrix representation: $B(x)=z(x)^{\mathrm{T}}Qz(x)$, where $z(x)$ is a power vector and $Q$ is a coefficient matrix. The trace $Tr(\cdot)$ of $Q$ is regarded as an approximation of the barrier function volume \cite{wang2018permissive}. Thus, enlarging the barrier function certified ROA can be approximated by enlarging the trace of $Q$ as shown in the following result. 
		
		\begin{theorem}\label{them:2}
			Given the learned control-affine system (\ref{eqn:learnsys}), suppose the input $u(x)$ is fixed. If there exists a sublevel set of Lyapunov function $\mathcal{L}_V = \{x\in \mathcal{L}_0, V(x)\leq c\}$ and two given SOS polynomials $s_1(x),s_2(x)$ generated from (\ref{eqn:them1}), then an optimal control barrier function $B(x)$ can be computed as follows
			\begin{equation}\label{eqn:them2}\small
				\begin{aligned}
					& &&\qquad \qquad \quad \underset{B(x)\in\mathcal{P};\;\; s_1(x),s_2(x)\in\mathcal{P}^{SOS}}{Q^*=\max \;Tr(Q)} \\
					&\text{s.t.} && -\frac{\partial V(x)}{\partial x}(P_k(x)+g(x)u+d)-s_1(x)B(x) \in \mathcal{P}^{\text{SOS}},\\
					& && \frac{\partial B(x)}{\partial x}(P_k(x)+g(x)u+d)+(\alpha -s_2(x))B(x) \in \mathcal{P}^{\text{SOS}},
				\end{aligned}
			\end{equation}
			\noindent where $Q$ is the coefficient matrix of $B(x)$ in square matrix representation, $\alpha$ is a linear factor to construct the extended $\kappa$ function.
		\end{theorem}
		\begin{proof}
			The dynamics of (\ref{eqn:learnsys}) is asymptotically stable with a fixed control input $u(x)$. Based on the P-satz from Lemma 4, the first constraint in (\ref{eqn:them2}) ensures that the derivative of the $V(x)$ in the compact set certified by $B(x)$ is non-decreasing, which guarantees the stabilization of (\ref{eqn:learnsys}). In addition, also followed by the P-satz, if the second constraint of (\ref{eqn:them2}) holds, it yields the condition (\ref{eqn:relaxB}), which implies that safety of (\ref{eqn:learnsys}) is guaranteed by the control input $u(x)$. Consequently, any state inside the safe region certified by $B^*(x)=Z(x)^{\mathrm{T}}Q^*Z(x)$ can never reach the unsafe region $\mathcal{X}\backslash\mathcal{B}^*$, while the region is also guaranteed to obtain the stabilization. Thus, we conclude that an optimal control barrier function $B^*(x)$ certifies the largest estimated ROA.
		\end{proof}	
		In practice, unsafe regions usually exist and have a big influence on controller design, e.g. the obstacles in autonomous driving. We describe the unsafe regions in the following form:
		\begin{equation}\label{eqn:unsaferegion}
			\begin{aligned}
				\mathcal{L}_{i=1,2,3,\dots}=\{\forall x \in\mathcal{X},m_{i=1,2,3,\dots}(x)\leq 0\}.
			\end{aligned}
		\end{equation}
		\noindent The following result reveals how we can find an optimal controller that maximize the ROA:
		\begin{theorem}
			Given the learned system (\ref{eqn:learnsys}) with a fixed $u(x)$ and a pre-defined Lyapunov sublevel set $\mathcal{L}_0=\{\forall x\in \mathcal{X}\backslash 0, 0< V(x)\leq c\}$, if there exists a control barrier function $B(x)$ that satisfies the following optimization,
			
			\begin{equation}\label{eqn:prop1}\small
				\begin{aligned}
					& && \quad \underset{B(x)\in\mathcal{P};\;\; s_1(x),s_2(x),n_1(x),n_2(x),n_3(x),\dots\in\mathcal{P}^{SOS}}{Q^*=\max \;Tr(Q)}\\
					&\text{s.t.} &&-\frac{\partial V(x)}{\partial x}(P_k(x)+g(x)u+d)-s_1(x)B(x)\in \mathcal{P}^{\text{SOS}},\\
					& && \frac{\partial B(x)}{\partial x}(P_k(x)+g(x)u+d)+(\alpha -s_2(x))B(x)\in \mathcal{P}^{\text{SOS}},\\
					& && -n_{i}(x)m_{i}(x)-B(x)\in \mathcal{P}^{\text{SOS}},
				\end{aligned}
			\end{equation}
			\noindent then the safety region $\mathcal{L}$ certified by the barrier function $B^*(x)=Z(x)^{\mathrm{T}}Q^*Z(x)$ can be maximized regarding the unsafe regions $\mathcal{L}_{i=1,2,3,\dots}$ in (\ref{eqn:unsaferegion}). 
		\end{theorem}
		\begin{proof}
			The proof is similar to the one of Theorem \ref{them:2}. we omit it due to limited space.
		\end{proof}
		
		Note that the above optimization is based on a fixed Lyapunov function. To further enlarge the estimate of ROA, We aim to search for an optimal Lyapunov function in the following result.
		
		\begin{proposition}
			Given a CBF $B(x)$ of the system (\ref{eqn:learnsys}), there exists an optimal CLF inside if it satisfies
			\begin{equation*}\small
				\label{eqn:optimal_V}
				\begin{aligned}\small
					& && \qquad \qquad \underset{\substack{ V\in\mathbb{R}^+, ~u(x)\in \mathcal{P},~ L(x)\in\mathcal{P}^\text{SOS}} }{\text{max}} V_\gamma\\
					& \text{s.t.} && V-L_1 B(x)\in \mathcal{P}^{\text{SOS}}\\
					&  &&  -\frac{\partial V(x)}{\partial x}(P_k(x)+g(x)u+d) - L_2(x)B(x) - V_\gamma \in \mathcal{P}^{\text{SOS}}\\
					& && -n_i(x)m_i(x)+V(x)\in \mathcal{P}^{\text{SOS}}.\\
				\end{aligned}
			\end{equation*}
		\end{proposition}
		
		\begin{proof}
			The ROA certified by an optimal CBF $\mathcal{L} = \{x\in\mathcal{X} \vert B(x)>0\}$ can be rewritten as $V(x)=c-B(x)$. According to P-satz, we can declare that inside the compact set $\mathcal{L}$, there always exist some auxiliary SOS polynomial such that the first and second constraint can be established, which can allow the sublevel set of $V(x)$ to generate itself inside a given safe region and without enter any unsafe regions. Note that, the time derivative of $B(x)$ is 
			
			$$\frac{\partial V(x)}{\partial x} \dot{x}=-\frac{\partial B(x)}{\partial x}\dot{x}$$ which is always negative in the sublevel set of $V(x)$ such that allow us to maximum the decision variable $V_\gamma$ to search for an optimal $V(x)$ with the largest derivative margin.
		\end{proof}
		
		\section{SOS-based Algorithm Development}
		An optimal Lyapunov-Barrier alteration algorithm is proposed to synthesis the controller $u(x)$ and the corresponding estimated ROA.  

		\subsection{Loop 1: Search for the maximum ROA with a given $V(x)$}
		
		Specify a Lyapunov function $V(x)$, and find its maximum sublevel set $c^*$ by using a bilinear search.
		
		\subsubsection{Step 1: Find an SOS polynomial $L(x)$} 
		Given a fixed $c_0$, then search for a feasible SOS polynomial $L(x)$,
		
		\begin{equation}
			\label{eqn:lp1-1}
			\begin{aligned}
				-\frac{\partial V(x)}{\partial x}\dot{x} - L(x)(c_0-V(x)) \in \mathcal{P}^{\text{SOS}}.
			\end{aligned}
		\end{equation}
		
		\subsubsection{Step 2: Find the maximum sublevel set of $V(x)$} 
		\begin{equation}
			\label{eqn:vcopt}
			\begin{aligned}
				& c^* = && \qquad \qquad \underset{\substack{ c\in\mathbb{R}^+, ~u(x)\in \mathcal{P},~ L(x)\in\mathcal{P}^\text{SOS}} }{\text{max}} c\\
				& \text{s.t.} &&  -\frac{\partial V(x)}{\partial x}\dot{x} - L(x)(c-V(x)) \in \mathcal{P}^{\text{SOS}}.
			\end{aligned}
		\end{equation}
		
		The two steps in Loop 1 can be repeated sequentially until $c^*$ stop increasing. Based on this bilinear search, we can set an initial barrier certificate as $\bar{h}(x)=c^*-V(x)$, which satisfies the definitions (\ref{Def:B1}) and (\ref{eqn:relaxB}), simultaneously.
		
		\subsection{Loop 2: Search for an optimal $B(x)$}
		
		In this loop, the P-satz and bilinear search method are employed to compute a permissive control barrier function, where it includes two steps:
		
		\subsubsection{Step 1: Fix $h(x)$, search for $u(x)$, $L_1(x)$, and $L_2(x)$}
		Using $h(x)$ obtained from previous step, we can search for feasible $u(x)$, $L_1(x)$, and $L_2(x)$.
		\begin{equation}
			\label{eqn:lp2-1}
			\begin{aligned}
				&& \underset{ \substack{\epsilon\geq 0,~u(x)\in\mathcal{P},~L_1(x),L_2(x)\in\mathcal{P}^{\text{SOS}}}}{\text{max}} \quad
				\epsilon \\
				&  \text{s.t.}
				& \hspace{-0.2in} -\frac{\partial V(x)}{\partial x}\dot{x} - L_1(x) h(x)&\in \mathcal{P}^{\text{SOS}},\\
				&
				& \hspace{-0.2in} \frac{\partial h(x)}{\partial x}\dot{x} - L_2(x) h(x) - \epsilon &\in \mathcal{P}^{\text{SOS}}.
			\end{aligned}
		\end{equation}
		
		\subsubsection{Step 2: Fix $u(x)$, $L_1(x)$, and $L_2(x)$, search for $h(x)$}
		
		In this step, the target of SOSP, $h(x)$, can be rewritten into the square matrix representation (SMR) form as $h(x)=Z(x)^TQZ(x)$, such that we can approximate the volume of ROA as the trace value of the coefficient matrix $Q$, for more details, we kindly recommend interested reader to \cite{wang2018permissive}. Based on the given $u(x)$, $L_1(x)$, $L_2(x)$, and unsafe regions $m_i(x),i\in\mathcal{M}$, an appropriate $h(x)$ can be computed as
		\begin{equation}
			\label{eqn:lp2-2}
			\begin{aligned}
				& \underset{ \substack{h(x)\in\mathcal{P}, L_1(x),L_2(x),n_i(x)\in\mathcal{P}^\text{SOS},i\in \mathcal{M}}}{\text{max}} &\quad
				\text{trace}(Q) \\
				&   \text{s.t.} -\frac{\partial V(x)}{\partial x}\dot{x} - L_1(x) h(x)&\in \mathcal{P}^{\text{SOS}},\\
				& \frac{\partial h(x)}{\partial x}\dot{x} - L_2(x) h(x) &\in \mathcal{P}^{\text{SOS}},\\
				& -n_i(x)m_i(x)-h(x) &\in \mathcal{P}^{\text{SOS}}.
			\end{aligned}
		\end{equation}
		
		\subsection{Loop 3: Search for an optimal $V(x)$} 
        In this loop, we attempt to find an optimal CLF $V(x)$ and further enlarge the estimated ROA. The optimal $V(x)$ could be found by using the following optimization,
		\begin{equation}
			\label{eqn:lp3-1}
			\begin{aligned}
				& \underset{ \substack{\epsilon_V\geq 0,,L_{V_1}(x),L_{V_2}(x)\in\mathcal{P}^{\text{SOS}}}}{\text{max}} & \epsilon_V \\
				\text{s.t.} & \qquad V - L_{V_1}(x)B(x)&\in \mathcal{P}^{\text{SOS}},\\
				& -\frac{\partial V(x)}{\partial x}\dot{x} -L_{V_2}B(x) - \epsilon_V &\in \mathcal{P}^{\text{SOS}}. 
			\end{aligned}
		\end{equation}
		
		\noindent Once we obtain the optimal solution $V^*(x)$, it can be transmitted to Loop 1, which begins another iteration of computation. 

		\section{Numerical Examples}
		\subsection{Example 1: A 2D Nonlinear System}
		Consider a two-dimensional system as follows,
		\begin{eqnarray}\label{demo:2D}\small
			\begin{aligned}
				\begin{bmatrix} \dot{x}_1 \\ \dot{x}_2 \end{bmatrix} = 
				\begin{bmatrix}
					-x_{1}+x_{2}+u_1\\	x_{1}^{2}x_{2}+1-\sqrt{\vert\exp(x_1)\cos(x_1)\vert}+u_2+d(x)
				\end{bmatrix},
			\end{aligned}
		\end{eqnarray}
		
		\noindent where $[x_1,x_2]^{\mathrm{T}}\in \mathbb{R}^2$ and $[u_1,u_2]^{\mathrm{T}}\in\mathbb{R}^2$ are the state and control input. Some unsafe regions are defined as
		\begin{equation}\label{eqn:unsafe_region_2d}\small
			\begin{aligned}
				q_1(x) &= (x_1+4)^2+(x_2-5)^2-4, \\
				q_2(x) &= (x_1+0)^2+(x_2+5)^2-4, \\
				q_3(x) &= (x_1-5)^2+(x_2-0)^2-5.   
			\end{aligned}
		\end{equation}
		
		The non-polynomial term in (\ref{demo:2D}) can be approximated by Chebyshev interpolants of degree $4$ in $[-2,2]$. An approximated system in the form of (\ref{eqn:learnsys}) can be obtained based on the GP of $d_\xi(x)$. We selected a mean function $m(x)=0$ and squared-exponential (SE) kernel defined as $k(x,x^{\prime})=\sigma_f^2 \exp{(-\frac{(x - x^{\prime})^2}{2l^2})}$, where $\sigma_f = \exp{(0.1)}$ is signal covariance and $l = \exp{(0.2)}$ is length scale. According to Lemma \ref{lem:polynomial_mean}, we generated a $4^{th}$ order polynomial mean function after $400$ epochs. The comparison of the results from the default mean function and the polynomial one is given in Fig. \ref{fig: the demo2d_comparison}.
		
		First, we use the collected trajectory information starting at $(-0.5,0.2)$ to construct the prior dataset, and then validate this generated polynomial mean function with the trajectory information started from $(-0.4,0.4)$. Both of these processes are sampled within $30s$ and the time step is $0.1s$. The corresponding root-mean-square errors are $9.03\times 10^{-5}$ and $2.89\times 10^{-4}$, respectively.
		
		\begin{figure}[ht] 
			\centering
			\includegraphics[width=0.9\linewidth]{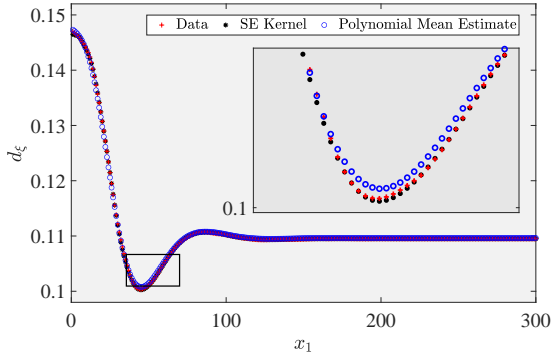}
			\caption{Example 1: Comparison of the prediction based on SE kernel and $4^{th}$ degree polynomial mean function. The green filled part denotes the variance of $[-\sigma,\sigma]$, the filled black point denotes the predictions and the red plus denotes the exact values. }
			\label{fig: the demo2d_comparison}
		\end{figure}
		
		\begin{figure}[ht] 
			\centering
			\includegraphics[width=0.9\linewidth]{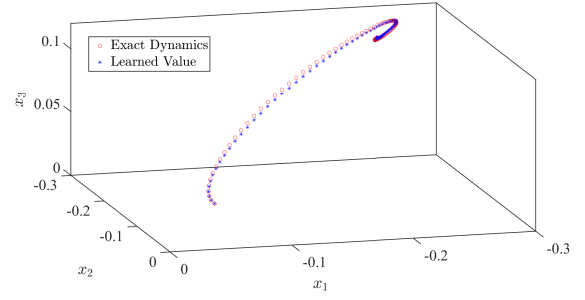}
			\caption{Example 1: Comparison of the learned value of the polynomial mean estimate and the observed values of these measurements. The former is in the form of the blue star while the latter is depicted by the red circle. }
			\label{fig:demo2d_comparison}
		\end{figure}

		\begin{figure}[ht] 
			\centering
			\includegraphics[width=0.9\linewidth]{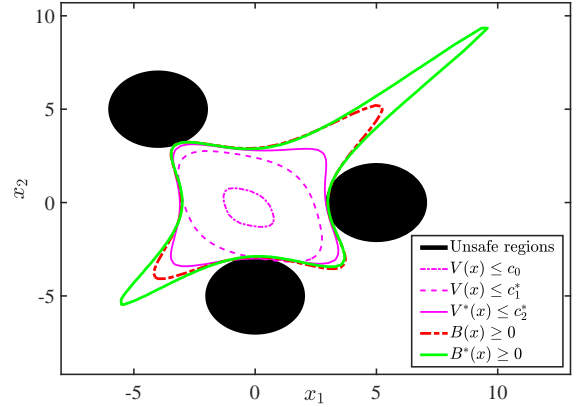}
			\caption{Example 1: Results comparison. The solid black ellipses depict three unsafe regions. The first three different magenta lines starting from the origin depict the initial sublevel set $c_0$, maximum sublevel set $c^*_1$ of $V(x)$ and the maximum sublevel set $c_2^*$ of $V^*(x)$. Generated from $V(x)\leq c^*_1$, the optimal barrier function $B(x)$ is drawn in the red dashed shape, while the optimal barrier function $B^*(x)$ derived from $V^*(x)$ (the proposed optimal Lyapunov-Barrier alteration algorithm) is circled in the solid green line.}
			\label{fig:demo_2d_cbf}
		\end{figure}

		Followed by the algorithm in Figure \ref{fig:algorithm}, the results in Fig. \ref{fig:demo_2d_cbf} display a comparison of safe stabilization estimation based on a Lyapunov-Barrier alteration method. In this case, it is clear that the sublevel set of $V(x)$ is included in $B^*(x)$ certified ROA, and the ROA certified by $\bar{B}^*(x)$ is obviously larger than the original $B^*(x)$.
		
		\subsection{Example 2: A 3D Nonlinear System}
		Consider a three-dimensional system as follows,
		\begin{eqnarray}\label{demo:3D}\footnotesize
			\begin{aligned}
				\begin{bmatrix} \dot{x}_1 \\ \dot{x}_2\\\dot{x}_3 \end{bmatrix} = 
				\begin{bmatrix}
					-x_{1}^2-\cos{(x_{1}^2)}\sin{(x_{1})}+u_1(x)+d_1(x)\\	
					-x_{2}-x_{1}^3x_{2}+u_2(x)\\
					-x_{1}^2x_{3}+1-\sqrt{|\exp{(x_{1})}\cos{(x_{1})}}+u_3(x)+d_3(x)
				\end{bmatrix}.
			\end{aligned}   
		\end{eqnarray}
		
		\noindent The unsafe regions are given as
		\begin{equation}\label{eqn:unsafe_region_3d}\small
			\begin{aligned}
				q_1(x) &= (x_1+4)^2+(x_2+4)^2+(x_3-4)^2-4, \\
				q_2(x) &= (x_1-0)^2+(x_2-4)^2+(x_3+0)^2-4, \\
				q_3(x) &= (x_1-4)^2+(x_2-0)^2+(x_3+4)^2-6.          
			\end{aligned}
		\end{equation}
		
		Nonlinear terms in (\ref{demo:3D}) are approximated by the Chebyshev interpolants in $[-5,5]\times [-5,5]\times [-5,5]$. The prior GP model of $d_1(x)$ and $d_3(x)$ share the same hyperparameters, including $m(x)=0$, $\sigma_f = 0.1$ and $l = 0.2$ of a SE kernel. We use the trajectory information starting from $(-0.1,0.1,0.1)$ to start this GP and validate these optimal hyperparameters with the trajectory information of $(-0.1,-0.2,0.1)$. Both of these processes are sampled within $30s$ and the sample time step is $0.05s$. After running over $1000$ epochs, the root-mean-square errors of among the SE kernel based value and the $4^{th}$ degree polynomial to the exact values are $0.00500$ and $0.00633$, respectively.
		
		The results in Fig. \ref{fig:demo_3d_Bar} show a comparison of estimations based on the Lyapunov function method and the proposed method. For this case, we can see that the sublevel set of the Lyapunov function is overlapping with the barrier function certified ROA (by using the proposed method). Again, the estimation of our method is significantly larger than the method of optimal Lyapunov function.
		
		\begin{figure}[ht] 
			\centering
			\includegraphics[width=1\linewidth]{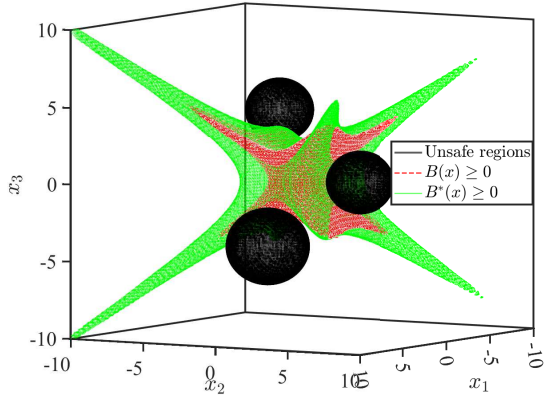}
			\caption{Example 2: Comparison of the CBF certified ROA $\mathcal{L}_B$ and $\mathcal{L}_{B^*}$ generated by the proposed method. The one in the red dashed line is generated by the maximum sublevel set of the optimal barrier function $B(x)$, while the one in the green solid line is generated by the maximum sublevel set of $B^*(x)$ from (\ref{eqn:lp3-1}) (the proposed optimal Lyapunov-Barrier alteration algorithm). The unsafe regions are depicted in the solid black balls.}
			\label{fig:demo_3d_Bar}
		\end{figure}
		
		\section{Conclusion}
		
        For partially unknown nonlinear systems, we first reconstruct a learned control affine system through Gaussian Processes and Chebyshev interpolants. Sufficient conditions based on sum-of-squares programs are proposed for the existence of a feasible controller such that the safety and stabilization of learned systems can both be guaranteed. Second, solvable conditions have proposed for pursuing an optimal control Lyapunov barrier function without and with considering pre-defined unsafe regions. Finally, an optimal Lyapunov-Barrier alteration algorithm is developed to compute the control input such that the estimated ROA can be maximized. Two numerical examples demonstrate that a significantly larger ROA can be obtained by the proposed method.


\end{document}